\documentclass[letterpaper, 10 pt, conference]{ieeeconf}  % Comment this line out
\usepackage{colortbl}
\definecolor{mygray}{gray}{.9}
\usepackage{cite}
\usepackage{color}
\usepackage{amsfonts,amssymb}
\usepackage{bm}
\usepackage{amsmath}
\usepackage{multirow}
\usepackage{algorithmic}
\usepackage{graphicx}
\usepackage{wrapfig}
\usepackage{subfigure}
\usepackage{mathtools}
\usepackage{float}
\usepackage{textcomp}
\usepackage{graphicx} %use graph format
\usepackage{epstopdf}
\usepackage{array}
\usepackage[thmmarks,amsmath]{ntheorem}
\newtheorem{mythm}{Theorem}
\newtheorem{mydef}{Definition}

\newtheorem{assumption}{Assumption}
\newtheorem{lemma}{Lemma}

\usepackage[title]{appendix}

\usepackage{array}
\usepackage{enumerate}
\usepackage{booktabs}
\usepackage{algorithm}
\usepackage{algorithmic}
\usepackage{setspace}
\usepackage{float}
\usepackage{cases}
\usepackage{mathrsfs}
\renewcommand{\QED}{\QEDopen}

\makeatletter\@mparswitchfalse\makeatother\normalmarginpar % to force the notes to the right
\usepackage[textwidth=1cm, textsize=scriptsize]{todonotes}

\IEEEoverridecommandlockouts                              % This command is only
\title{\bf A Functional Learning Approach for Team-Optimal Traffic Coordination
}

\author{Weihao Sun$^1$, Gehui Xu$^2$, 
Alessio Moreschini$^3$, Thomas Parisini$^{3,4}$, and Andreas A. Malikopoulos$^{1,2,5}$ 
\thanks{This work is supported in part by NSF under Grants CNS-2401007, CMMI-2348381, IIS-2415478, and in part by MathWorks.}
\thanks{$^1$W. Sun and A. A. Malikopoulos are with the Systems Engineering Program, Cornell University, Ithaca, NY 14850 USA. {\tt\small email: \{ws493,amaliko\}@cornell.edu}}
\thanks{$^2$G. Xu and A. A. Malikopoulos are with the School of Civil and Environmental Engineering, Cornell University, Ithaca, NY 14850 USA. {\tt\small email: \{gx62,amaliko\}@cornell.edu}}
\thanks{$^3$
A. Moreschini and T. Parisini are with the Department of Electrical and Electronic Engineering, Imperial College London, London SW7 2AZ, UK. {\tt\small email: \{
a.moreschini,\allowbreak
t.parisini\}@imperial.ac.uk}}
\thanks{$^4$T. Parisini is also with the Department of Electronic Systems, Aalborg University, Denmark and with the Department of Engineering and Architecture, University of Trieste, Italy. {\tt\small email: t.parisini@imperial.ac.uk}}
\thanks{$^5$A. A. Malikopoulos is with the Applied Mathematics, Systems Engineering, Mechanical Engineering, Electrical \& Computer Engineering, and School of Civil \& Environmental Engineering, Cornell University, Ithaca, NY, USA. (email: \texttt{amaliko@cornell.edu}) }
}

\begin{document}

\maketitle
\thispagestyle{empty}
\pagestyle{empty}

%%%%%%%%%%%%%%%%%%%%%%%%%%%%%%%%%%%%%%%%%%%%%%%%%%%%%%%%%%%%%%%%%%%%%%%%%%%%%%%%
\begin{abstract}
In this paper, we develop a kernel-based policy iteration functional learning framework for computing team-optimal strategies in traffic coordination problems. We consider a multi-agent discrete-time linear system with a cost function that combines quadratic regulation terms and nonlinear safety penalties. Building on the Hilbert space formulation of offline receding-horizon policy iteration, we seek approximate solutions within a reproducing kernel Hilbert space, where the policy improvement step is implemented via a discrete Fréchet derivative. We further study the model-free receding-horizon scenario, where the system dynamics are estimated using recursive least squares, followed by updating the policy using rolling online data. The proposed method is tested in signal-free intersection scenarios via both model-based and model-free simulations and validated in SUMO.
\end{abstract}

%%%%%%%%%%%%%%%%%%%%%%%%%%%%%%%%%%%%%%%%%%%%%%%%%%%%%%%%%%%%%%%%%%%%%%%%%%%%%%%%
\section{INTRODUCTION}

% \tp{We need to be a bit careful and tone down a bit. In this paper we are considering a very simplified team optimal problem because each agent has the entire state information available}
Team coordination problems arise in many multi-agent systems in which several decision makers must act based on shared system information while pursuing a common objective \cite{radner1962team,Malikopoulos2021}. In such settings, the collective outcome depends both on the individual optimal controls and how the actions are coordinated over time under the underlying information structure. When the agents are fully cooperative, the relevant solution concept is the team-optimal solution, which is a joint policy profile minimizing the common objective over the admissible policy class~\cite{ho1972team,zoppoli2020neural}. In these problems, one seeks feedback policies that account for interactions among agents while remaining computationally tractable over a finite horizon.

Traffic coordination provides a representative and practically important class of team decision problems. In signal-free traffic scenarios such as intersection crossing \cite{Malikopoulos2020, liu2024intersection}, lane changing\cite{sun2025recommendation,xu2025game,wang2025corra}, and other coordinated maneuvers, each vehicle affects the motion and safety margin of the others. One representative scenario is the coordination of multiple connected and automated vehicles (CAVs) at a signal-free intersection, where vehicles are able to cooperate with each other to cross safely and efficiently \cite{Malikopoulos2020,liu2024intersection}.
%Motivated by this setting, this paper focuses on signal-free intersection coordination as a concrete scenario through which the proposed method can be formulated and evaluated.

To compute team-optimal feedback policies, this paper adopts a policy iteration functional learning framework.
% \tp{We need to say that our goal is to compute the coordination policy off-line which is the justification for using a functional approach. It is not clear at the moment} 
Repeatedly solving a finite-horizon optimal control problem may be computationally demanding. A common alternative is to improve a parameterized feedback policy through policy gradient methods \cite{Makila1987GP,fazel2018global, hu2023toward}. In existing research efforts, policy gradient methods have been applied to data-driven controls on the linear quadratic regulator \cite{zhao2023data}. Other results show the possibility to analyze the convergence theory of policy gradient with robustness guarantee \cite{zhang2021policy}. Also, to apply to similar traffic problems involving CAVs, some existing work utilizes policy optimization methods to address safety issues \cite{zhao2023multi, wei2019mixed}. However, standard policy gradient approaches are often sensitive to the choice of learning rate, which may affect both convergence and approximation accuracy \cite{Alessio2024Fretchet}. %Moreover, the standard policy gradient method remains dependent on the parameterization choice of the policy, which could introduce approximation errors. 
Recent efforts~\cite{moreschini2024non} considered a policy update step characterized via a discrete Fréchet derivative~\cite{Alessio2024Fretchet} of the finite-horizon cost, leading to independence from the learning rate choice and guaranteeing a monotonically non-increasing cost across iterations.

%Furthermore, in this work, the viewpoint is extended to a team coordination problem, and approximate optimal policies are constructed in a reproducing kernel Hilbert space (RKHS) through finite kernel expansions.

%The problem addressed in this paper is therefore to compute approximate team optimal feedback policies for a finite-horizon traffic coordination problem under both known and unknown system dynamics with a concrete signal-free intersection scenario through which the proposed method can be formulated and evaluated. 
In this paper, we address the problem of computing team-optimal feedback policies for a finite-horizon traffic coordination problem in signal-free intersection scenarios involving CAV coordination and potential human-driven vehicles (HDVs).
% \tp{What do we mean by "unknown dynamics"? The dynamics of the vehicles is known (double integrators in Eq. (4))}
% \textcolor{blue}{In the unknown-dynamics setting, we consider the scenario that includes a human-driven vehicle (HDV) whose driving response is unknown and interacts with the controlled CAVs.}
Following~\cite{moreschini2024non,Alessio2024Fretchet}, we develop an offline policy iteration algorithm where the policy improvement step is implemented through an implicit Fréchet derivative update. We show that the resulting sequence of policies yields a non-increasing finite-horizon cost and converges to a local team-optimal solution. We further provide a finite kernel-based approximation, implemented via a reproducing kernel Hilbert space (RKHS) \cite{berlinet2011reproducing}, along with a complexity analysis. Then, we extend this framework to an online receding-horizon setting in which the unknown dynamics are estimated through recursive least squares (RLS) followed by the policy updating over rolling horizon windows. We show that the cost decreases monotonically within each update window, and the algorithm admits a suboptimal team solution while reducing computational complexity. Finally, we demonstrate the proposed approach in intersection coordination scenarios under both model-based and model-free settings, and validate the learned controller in SUMO. 

The remainder of the paper is organized as follows. 
In Section~II, we introduce the problem formulation, including the dynamics and the optimal control problem. 
In Section~III, we present the model-based offline learning and, in Section~IV, we provide the model-free online learning framework. 
In Section~V, we provide a simulation experiment using SUMO to validate the effectiveness of the proposed framework. 
Finally, in Section~VI, we draw concluding remarks and discuss potential directions for future research.

% In the offline case, a full horizon policy update scheme is constructed where it is refined across the full horizon through an implicit improvement step. The resulting sequence of policies is analyzed to establish monotonic improvement of the associated cost under functional update. In the online case, the same framework is adapted to a receding-horizon structure, where the policy updates are performed over shorter rolling horizon windows, after estimating the system dynamics via recursive least squares (RLS). 

%Literature Reviews here...

% The main contributions of this paper are as follows. First, we formulate the signal-free traffic coordination as a finite-horizon team-optimal control problem with a nonlinear cost and cast the associated feedback design problem in a functional learning framework. Based on this formulation, we develop a kernel-based offline policy iteration algorithm where the policy improvement step is implemented through an implicit Fréchet derivative update. Then, we extend this framework to an online receding-horizon setting in which the unknown dynamics are estimated through RLS followed by the policy updating over rolling horizon windows. Finally, we demonstrate the proposed approach in intersection coordination scenarios under both model-based and model-free settings, and validate the learned controller in SUMO. 

\section{Team Coordination}
\subsection{Notation}
Let $\mathbb{R}$ and $\mathbb{N}$ denote the set of real and natural numbers, respectively. For $n,m \in \mathbb{N}$, $\mathbb{R}^n$ denotes the space of real column vectors of dimension $n$, and $\mathbb{R}^{n \times m}$ denotes the space of real matrices of dimension $n \times m$. For a vector $x \in \mathbb{R}^n$, $\|x\|$ denotes the Euclidean norm, and for a matrix $A$, $A^\top$ denotes its transpose. Given two matrices $A_1,A_2\in\mathbb{R}^{n \times n}$, 
the notation $A_1 \succ A_2$ ($A_1 \succeq A_2$) indicates that $A_1 - A_2$ is positive (semi-)definite. The discrete time index is denoted by $t \in \{0,1,\dots,T\}$, where $T \in \mathbb{N}$ is the finite horizon length. 
%State-feedback control policies are denoted by $\pi_t : \mathbb{R}^n \rightarrow \mathbb{R}^m$. 
The policies are assumed to belong to a Hilbert space $\mathcal{H}$ equipped with inner product $\langle \cdot, \cdot \rangle_{\mathcal H}$ and induced norm $\|\cdot\|_{\mathcal H}$. Expectation with respect to a random variable $x$ distributed according to $p$ is denoted by $\mathbb{E}_{x \sim p}[\cdot]$. For Monte Carlo approximations, $\{x^{(i)}\}_{i=1}^{N}$ denotes a set of i.i.d. samples drawn from $p$.

\subsection{Problem Formulation}

Consider a team problem with $N\in\mathbb{N}$ team members, indexed by $\mathcal{N}=\{1,\dots,N\}$. Each team member $i\in\mathcal{N}$ determines its control action $u_{i,t}$ at time $t$. 
The team collectively controls the discrete-time linear system
\begin{equation}\label{eq:linear_sys}
x_{t+1} = A x_t + \sum_{i=1}^N B_i u_{i,t}, 
\qquad t = 0,1,\dots,T-1,
\end{equation}
where 
$x_t \in \mathbb{R}^n$ is the state and 
$u_{i,t}$ is the control input of team member $i$. 
% \textcolor{blue}{where $\mathcal{U}_i$ is nonempty.}
% \tp{See comment just after eq. (8)}
Define the stacked input
$
u_t = \operatorname{col}(u_{1,t},\dots,u_{N,t}) \in \mathbb{R}^{m}$, $  m=\sum_{i=1}^N m_i,
$
and the aggregate input matrix
$
B = [B_1 \ \cdots \ B_N].
$
Then \eqref{eq:linear_sys} can be written compactly as 
\begin{equation}\label{eq:linear_sys_2}
    x_{t+1} = A x_t + B u_t.
\end{equation}

The initial state $x_0$ is a random variable with a known distribution $p_0$.
% \tp{I guess we should assume $p_0$ to be known, right?}
All team members independently determine their strategies and collaborate in minimizing a shared finite-horizon cost function given by
%We consider the finite-horizon cost function
%\begin{equation}
\begin{align}\label{eq:FH_cost}
J_{\mathrm{FH}}(x_0,u_{0:T-1})
&=
\mathbb{E}
\Bigg[
\sum_{t=0}^{T-1}
\Big(
x_t^\top Q x_t
+
u_t^\top R u_t
+
\psi(x_t)
\Big) \notag\\
&\qquad
+
x_T^\top Q_F x_T
+
\psi_F(x_T)
\Bigg],
\end{align}
where 
$Q,Q_F,R \succ   \!0$ are symmetric positive definite  matrices. 
The nonlinear functions 
$\psi:\mathbb{R}^n \!\to \!\mathbb{R}$ and 
$\psi_F:\mathbb{R}^n \!\to \!\mathbb{R}$ 
are continuously differentiable, bounded from below, and satisfy
$
\psi(0)\!=\!0$, $  \psi_F(0)\!=\!0.
$
The objective of the team is to find control policies 
$\{u_{i,t}\}_{i\in\mathcal{N},\, t=0}^{T-1}$
minimizing \eqref{eq:FH_cost}.

Common cost functions such as~\eqref{eq:FH_cost} can represent a wide range of traffic coordination scenarios. 
%Modern intelligent transportation systems involve multiple interacting vehicles operating in a shared and dynamic environment. 
% Typical applications include signal-free intersection coordination~\cite{Malikopoulos2020,liu2024intersection}, lane changing~\cite{sun2025recommendation}, vehicle merging~\cite{venkatesh2023connected}, and other coordinated maneuvers in mixed-autonomy traffic.
We focus on the signal-free intersection coordination problem as a representative scenario and provide a concrete formulation.
% Modern intelligent transportation systems involve multiple interaction vehicles operating in a shared and dynamic environment. Typical application scenarios in transportation include signal free intersection crossing coordination~\cite{Malikopoulos2020,liu2024intersection}, lane changing~\cite{sun2025recommendation}, vehicle merging~\cite{venkatesh2023connected}, and other coordinated maneuvers in mixed autonomy settings. In this paper, we focus on the signal free intersection coordination scenario. 

%\subsection{Policy Parameterization}

For each vehicle, we consider the following dynamics:
% \tp{Say that CT dynamics will be discretized}
\begin{equation}\label{eq:vehicle_dynamic}
\begin{aligned}
\dot{p}_i(t) = v_i(t), \quad
%\\
\dot{v}_i(t) = u_i(t),
\end{aligned}
\end{equation}
where {$p_i(t) \in \mathbb{R}^2$} denotes the position, $v_i(t) \in \mathbb{R}$ denotes the speed, and $u_i(t)\in \mathbb{R}$ denotes the acceleration, which serves as the control input. We discretize the continuous-time dynamics in our problem.

The team seeks to minimize a common performance objective that captures tracking performance, control effort, and collision avoidance, \emph{i.e.}, 
\begin{align}\label{eq:objective_intersection}
\min_{u_{0:T-1}} \;
J
&=
\sum_{t=0}^{T-1}
\Bigg(
\sum_{i \in \mathcal{N}}
\big(
x_{i,t}^\top Q x_{i,t}
+
u_{i,t}^\top R u_{i,t}
\big)
+
\Phi(x_t)
\Bigg)\notag \\
&\quad
+
\sum_{i \in \mathcal{N}}
x_{i,T}^\top Q_F x_{i,T}
+
\Phi(x_T).
\end{align}
The function $\Phi(\cdot)$ represents a collision-avoidance penalty and serves as a soft safety constraint \cite{liu2024intersection}. 
To avoid double-counting pairwise interactions, it is defined  as
\begin{equation}\label{eq:team_collision_penalty}
\Phi(x_t)
\coloneq 
\sum_{1 \le i < j \le N}
\frac{d_d^2}{d_{ij}(t)^2 + \delta},
\end{equation}
where $d_d > 0$ is the desired safety distance and $\delta>0$ is a small constant. The cost \eqref{eq:team_collision_penalty} penalizes small inter-vehicle distances 
\begin{equation}\label{eq:dij}
d_{ij}(t) = \|p_i(t) - p_j(t)\|_2.
\end{equation}
%with $p_i(t)$ denoting the position of vehicle $i$.
%We now define the objective function in the specific intersection coordination scenario as follows:
% The common objective function is given by
% \begin{equation}\label{eq:objective_intersection}
% \begin{aligned}
% \min_{u_{i,0:T-1}} \; 
% J_i =
% \sum_{t=0}^{T-1}
% \Big(
% x_{i,t}^\top Q x_{i,t}
% + u_{i,t}^\top R u_{i,t}
% + \psi_i(x_t)
% \Big)
% \\ + x_{i,T}^\top Q_F x_{i,T} + \psi_{i,F}(x_T),
% \end{aligned}
% \end{equation}
% We add the penalty on collision as a soft constraint defined by $\psi_i(x_t)$:

% \begin{equation}\label{eq:psi}
% \psi_i(x_t)
% =
% \sum_{j \in \mathcal{N}_i}
% \frac{d_d^2}{d_{ij}(t)^2 + \delta},
% \end{equation}

% where 
% \begin{equation}\label{eq:dij}
% d_{ij}(t)
% =
% \sqrt{\big(p_i(t)-p_j(t)\big)^2},
% \end{equation}

%\subsection{Stage-wise Cost Functional}

% Following a stage-wise policy iteration scheme, we define the policy improvement functional
% \begin{equation}\label{eq:Jte}
% \begin{aligned}
% J_t^e(\pi_{t:T-1})
% &=
% \mathbb{E}
% \Big[
% x_t^\top Q x_t
% +
% \pi_t(x_t)^\top R \pi_t(x_t)
% +
% \psi(x_t)\\
% &\quad+
% V_{t+1}(A x_t + B \pi_t(x_t))
% \Big].
% \end{aligned}
% \end{equation}

We consider all team members making decisions under a closed-loop information structure, \emph{i.e.},  they choose instantaneous control strategies 
$u_{i,t}(x_t)$  based on the observed state~$x_t$. 
For each  $i \in \mathcal{N}$, let 
$
\pi_{i,t}:  \mathbb{R}^{n} \to \mathbb{R}^{m_i}
$
denote a state-feedback control policy at time $t$. Assume that $\pi_{i,t}$ belongs to a function space $\mathcal{H}_i$, where $\mathcal{H}_i$ is a Hilbert space of measurable functions from $\mathbb{R}^n$ to $\mathbb{R}^{m_i}$. 
Define the joint policy at time $t$ as
$
\pi_t = \mathrm{col}\{\pi_{1,t}, \dots, \pi_{N,t}\}\in \mathcal{H},
$ where $\mathcal{H} \coloneq \prod_{i=1}^N \mathcal{H}_i$, 
and the policy sequence
$
\pi_{0:T-1} \coloneq  \{\pi_0, \dots, \pi_{T-1}\}.
$
Following a policy-based formulation, the expected finite-horizon cost is defined as
\begin{equation}\label{eq:expected_cost}
J(\pi_{0:T-1})
\!\coloneq \!
\mathbb{E}_{x_0 \sim p_0}\!
\left[
J_{\mathrm{FH}}(x_0,u_{0:T-1})
\right],
u_{i,t} \!=\! \pi_{i,t}(x_t),  \!i\! \in\! \mathcal{N}.
\end{equation}
% Let  $\mathcal{T} \coloneq  \{0,\dots,T-1\}$ be the time sequence and $\Pi$ denote the set of admissible (state-feedback) 
% %\tp{I'm not sure our functional methodology can handle constrained functional optimization as yet. We need to add soft penalties on $u_t$}
% policy sequences defined as follows:
% \[
% \Pi\! \coloneq \!
% \Big\{
% \pi_{0:T-1} \!= \!\{\pi_0,\dots,\pi_{T-1}\}
%  \Big| 
% \pi_{i,t}\!:\!\mathbb{R}^n \!\to\! \mathcal{U}_i,
%  i \in \mathcal{N}, t \!\in\! \mathcal{T}
% \Big\}.
% \]
The team outcome resulting from the joint decisions is characterized by the team-optimal solution~\cite{radner1962team,zoppoli2020neural,Malikopoulos2021}. This solution corresponds to a strategy profile under which no unilateral or joint deviation by team members can yield improved collective performance~\cite{xu2025does,Malikopoulos2021},
\begin{mydef}[Team-optimal Solution]
A policy sequence $
\pi_{0:T-1}^\star
$
is called a team-optimal solution if
\[
J(\pi_{0:T-1}^\star) \le J(\pi_{0:T-1}),
\qquad \forall \pi_{0:T-1} \in \prod_{t=0}^{T-1} \mathcal{H}.
\]
\end{mydef}
%The team-optimal solution characterizes the cooperative behavior of all members under full collaboration.

In this paper, the objective is to find the team-optimal policy $\pi_{0:T-1}^\star$ that minimizes the cost \eqref{eq:FH_cost}, subject to the system dynamics in \eqref{eq:linear_sys_2}.
%find solution of the team-optimal problem defined in Problem 1.

% \begin{problem}\label{team_prob}
% Find the finite-horizon team optimal policy $\{\pi^\star_t\}, t=0,\cdots,T-1$, and such that minimizes the cost \eqref{eq:FH_cost}, subject to the system dynamic in \eqref{eq:linear_sys_2}.
% \end{problem}

\section{Model-based offline learning}

% In this paper, we first extend the offline functional learning algorithm in \cite{} to a model-based traffic coordination setting with known dynamics, and specialize it to cross intersection scenarios considering safety. 

In this section, we propose an offline approach to compute the closed-loop optimal control policy. The system dynamics are assumed to be known, and the policy is improved through the full horizon in each iteration via an implicit Fréchet derivative update~\cite{moreschini2024non}.
% Due to the presence of the nonlinear state cost terms $\psi$ and $\psi_F$,
% the optimal policy is generally nonlinear.
% In this paper, we seek approximate solutions within a linearly parameterized (possibly nonlinear) policy class.

\subsection{Policy evaluation and improvement}

To compute the sequence of policies, we adopt a policy iteration framework~\cite{bertsekas2012dynamic}. The policy iteration consists of two steps at each iteration: policy evaluation and improvement, repeated until convergence.

For each stage $t$, we define the cost-to-go function
\begin{equation}
V_T(x) \coloneq  x^\top Q_F x + \psi_F(x),
\end{equation}
and for $t=T-1,\dots,0$,
\begin{align}\label{eq:cost_to_go1}
V_t(x_t)
&=
x_t^\top Q x_t
+
\pi_t(x_t)^\top R \pi_t(x_t)
+
\psi(x_t)\notag
\\
&\quad+
V_{t+1}(A x_t + B \pi_t(x_t)).
\end{align}

\noindent\textbf{{Policy Evaluation}.}
Let $\pi_{0:T-1} = \{\pi_0, \dots, \pi_{T-1}\}, \pi_t \in \prod_{t=0}^{T-1} \mathcal{H}$ be a given sequence of feedback policies in a Hilbert space.  The forward dynamics for \eqref{eq:linear_sys}  with $\pi_{0:T-1}$ are simulated for all initial states $x_0 \sim p_0$, evaluating $V_0,\dots,V_{T-1}$ along the state trajectories.

\noindent\textbf{Policy Improvement.}
The policy sequence $\pi_{0:T-1}$ is improved recursively backwards through the time horizon by minimizing the cost functional $\tilde{J_t}:\mathcal{H}^{T-t}\rightarrow \mathbb{R}$ defined by
% \begin{equation}\label{eq:piimprove_cost}
% \tilde{J_t}(\pi_{t:T-1}) = \mathbb{E}_{x_t\sim p_t}\Big[h(x_tm\pi_t(x_t)) + V_{t+1}\big(f(x_t,\pi_t(x_t))\big) \Big]
% \end{equation}
\begin{equation}\label{eq:piimprove_cost}
\begin{aligned}
\tilde{J_t}(\pi_{t:T-1})
&=
\mathbb{E}
\Big[
x_t^\top Q x_t
+
\pi_t(x_t)^\top R \pi_t(x_t)
+
\psi(x_t)\\
&\quad+
V_{t+1}(A x_t + B \pi_t(x_t))
\Big].
\end{aligned}
\end{equation}

\subsubsection{Fréchet Derivative in Hilbert Space}
In the existing literature, the functional gradient descent method is one of the most common approaches for minimizing \eqref{eq:piimprove_cost} by seeking the optimal solution along the gradient of the cost functional. However, this approach is sensitive to the learning rate and is difficult to fine-tune to ensure convergence in Hilbert spaces for the algorithm. Therefore, as proposed in \cite{Alessio2024Fretchet}, a discrete Fréchet derivative method is adapted to the algorithm, which is capable of addressing this gap.

\begin{mydef}[Fréchet Differentiable]
A function $S:\mathcal{H}\rightarrow\mathbb{R}$ is called Fréchet differentiable at $\phi\in\mathcal{H}$ if there exists a $dS(\phi)\in\mathcal{H}$ such that 
\begin{equation}
\lim_{\|\delta \phi\| \to 0^+}
\frac{
|S(\phi + \delta \phi) - S(\phi) - dS(\phi)(\delta \phi)|
}{
\|\delta \phi\|
}
= 0.
\end{equation}
\end{mydef}

\begin{mydef}[Fréchet Derivative]\label{def:Fréchet Derivative}
Let $S : \mathcal{H} \to \mathbb{R}$ be a function that is Fréchet differentiable on $\mathcal{H}$. For any $\phi, \psi \in \mathcal{H}$ and the Riesz representation of the Fréchet derivative $DS(\phi)\in\mathcal{H}$, the Fréchet derivative $\overline{DS}:\mathcal{H}\times\mathcal{H}\rightarrow\mathcal{H}$ satisfies
\begin{equation}
\langle \phi - \psi, \overline{DS}(\phi,\psi) \rangle
=
S(\phi) - S(\psi),
\end{equation}
\begin{equation}
\lim_{\|\phi-\psi\|\to 0}
\overline{DS}(\phi,\psi)
=
DS(\psi).
\end{equation}
\end{mydef}

\subsubsection{Policy Iteration by Discrete Fréchet Derivatives}
We now specialize the discrete Fréchet framework to the stage-wise policy improvement functional defined in \eqref{eq:piimprove_cost}. We denote $\pi_{0:T-1}^k 
\coloneq  
\{\pi_0^k, \ldots, \pi_{T-1}^k\}$ as the policy sequence at iteration $k$. Our objective is to construct an updated sequence $\pi^{k+1}_{0:T-1}$ such that the total expected cost decreases monotonically, \emph{i.e.},

\begin{equation}\label{eq:cost_diff}
\tilde{J}_0(\pi_{0:T-1}^{k+1})
-
\tilde{J}_0(\pi_{0:T-1}^k)
\le 0.
\end{equation}

Using the backward decomposition implied by \eqref{eq:cost_to_go1} and~\eqref{eq:piimprove_cost}, the total cost difference in \eqref{eq:cost_diff} can be given by
\begin{align}\label{eq:cost_diff_eq}
&\sum_{t=0}^{T-2}
\Big(\widetilde{J_{t}}(\pi_t^{k+1}, \pi^{k+1}_{t+1:T-1})
\Big) + \widetilde{J}_{T-1}(\pi^{k+1}_{T-1})\notag \\
&\qquad -\sum_{t=0}^{T-2}
\Big(\widetilde{J_{t}}(\pi_t^{k}, \pi^{k+1}_{t+1:T-1})
\Big) - \widetilde{J}_{T-1}(\pi^{k}_{T-1}) \le 0.
\end{align}
By combining the summation, we have
\begin{align}\label{eq:cost_diff_eq2}
&\sum_{t=0}^{T-2}
\Big(\widetilde{J_{t}}(\pi_t^{k+1}, \pi^{k+1}_{t+1:T-1}) -
\widetilde{J_{t}}(\pi_t^{k}, \pi^{k+1}_{t+1:T-1})
\Big) \notag \\
&\qquad + \widetilde{J}_{T-1}(\pi^{k+1}_{T-1}) - \widetilde{J}_{T-1}(\pi^{k}_{T-1}) \le 0.
\end{align}
Therefore, it is sufficient to ensure that for every stage $t$, 
\begin{equation}\label{eq:J_cost_pi}
\widetilde J_t(
\pi_t^{k+1},
\pi_{t+1:T-1}^{k+1}
)
-
\widetilde J_t(
\pi_t^{k},
\pi_{t+1:T-1}^{k+1}
)
\le 0.
\end{equation}

Then, by Definition~\ref{def:Fréchet Derivative}, we can rewrite the left hand side of equation \eqref{eq:J_cost_pi} as
\begin{align}
\widetilde J_t(
\pi_t^{k+1},&
\pi_{t+1:T-1}^{k+1}
)
-
\widetilde J_t(
\pi_t^{k},
\pi_{t+1:T-1}^{k+1}
)\notag\\
=&\langle\pi_{t}^{k+1}-\pi_{t}^{k},\overline{D\tilde{J_{t}}}(\pi_{t}^{k+1},\pi_{t}^{k},\pi_{t+1:T-1}^{k+1})\rangle,
\end{align}
where $\overline{D\tilde{J_{t}}}(\cdot,\cdot,\pi_{t+1:T-1}^{k+1}):\mathcal{H}\times\mathcal{H}\rightarrow\mathcal{H}$ denotes the discrete Fréchet derivative corresponding to each pair of cost differences in \eqref{eq:cost_diff_eq2}. We then adopt the implicit update rule in which the policies $\pi_{t+1:T-1}^{k+1}$ are updated by taking a step in the direction of the discrete Fréchet derivative:
\begin{equation}\label{eq:implicit_update}
\pi_{t}^{k+1}=\pi_{t}^{k}-\delta\overline{D\tilde{J}_{t}}\big(\pi_{t}^{k+1},\pi_{t}^{k},\pi_{t+1:T-1}^{k+1} \big),
\end{equation}
where $\delta>0$ is the learning rate.

On this basis, following the lines of~\cite{moreschini2024non}, we can establish the convergence guarantee of \eqref{eq:implicit_update} to a local team-optimal solution.
% {  On this basis, we can establish the convergence guarantee of \eqref{eq:implicit_update} to a local minimum.   In particular, any accumulation point of the sequence ${\pi_{0:T-1}^{k}}$ is a local minimum, which satisfies the first-order necessary condition for the team-optimal policy $\pi^\star_{0:T-1}$.
%a consistent decrease in the cost function evaluation with each iteration until convergence
% In particular, any accumulation point of the sequence ${\pi_{0:T-1}^{k}}$ is a stationary point, which satisfies the first-order necessary condition for the team-optimal policy $\pi^\star_{0:T-1}$.

\begin{mythm}\label{thm1}
Consider the policy iteration defined by the cost functional $\tilde{J}_t$ in \eqref{eq:piimprove_cost}. Let the policy update at each iteration $k$ be governed by the implicit rule \eqref{eq:implicit_update} with learning rate $\delta>0$. Then,  for every iteration $k$ and any $\delta>0$, %at each time step $t$,
%the algorithm generates a monotonically decreasing sequence of the $t$-th cost $\widetilde{J_{t}}(\pi_t^{k+1}, \pi^{k+1}_{t+1:T-1})$. Moreover, 
\begin{align}\label{eq:thm1}
&\widetilde{J}_0(\pi_{0:T-1}^{k+1}) - \widetilde{J}_0(\pi_{0:T-1}^{k}) = \notag\\
&\qquad-\frac{1}{\delta} \sum_{t=0}^{T-1} \| \pi_{t}^{k+1} - \pi_{t}^{k} \|_{\mathcal{H}}^{2} \le 0.
\end{align}
% and the total expected cost i monotonically non-increasing for any $\delta>0$.
% for any $\delta$, with vanishing increments on aggregated policy $\sum_{t=0}^{T-1} \| \pi_{t}^{k+1} - \pi_{t}^{k} \|_{\mathcal{H}}^{2} \to0$ as $k\to\infty$.     
Therefore, the sequence $\{\widetilde J_0(\pi_{0:T-1}^{k})\}_{k\ge0}$ is monotonically non-increasing
%. If $\widetilde J_0$ is lower bounded, then 
and converges as $k\to\infty$, and
$
\sum_{t=0}^{T-1}\|\pi_t^{k+1}-\pi_t^k\|_{\mathcal H}^2 \to 0 .
$
\end{mythm}

\begin{proof}
From \eqref{eq:cost_diff} to \eqref{eq:implicit_update}, we can express the cost difference at each time step $t$ as
\begin{align}
&\widetilde J_t\!\left(
\pi_t^{k+1},\pi_{t+1:T-1}^{k+1}
\right)
-
\widetilde J_t\!\left(
\pi_t^{k},\pi_{t+1:T-1}^{k+1}
\right) \notag\\
&\qquad\qquad=
-\frac{1}{\delta}
\left\|
\pi_t^{k+1}-\pi_t^{k}
\right\|_{\mathcal H}^2 \leq0,
\end{align}
which implies that the sequence of the t-th cost $\widetilde J_t\!\left(
\pi_t^{k+1},\pi_{t+1:T-1}^{k+1}
\right)$ is monotonically decreasing.
Therefore, by summing through $t$ we can obtain the left-hand side of \eqref{eq:cost_diff_eq2} as 
\begin{align}\label{eq:thm1_proof}
&\widetilde{J}_0(\pi_{0:T-1}^{k+1}) - \widetilde{J}_0(\pi_{0:T-1}^{k}) = \notag\\
&\qquad-\frac{1}{\delta} \sum_{t=0}^{T-1} \| \pi_{t}^{k+1} - \pi_{t}^{k} \|_{\mathcal{H}}^{2} \le 0.
\end{align}
Then, we sum the total cost difference over iteration $k=0,1,...,K$, we have 
\begin{equation}
\frac{1}{\delta}
\sum_{k=0}^{K}
\sum_{t=0}^{T-1}
\left\|
\pi_t^{k+1}-\pi_t^{k}
\right\|_{\mathcal H}^2
=\tilde{J}_0(\pi^0_{0:T-1})-\tilde{J}_0(\pi^{k+1}_{0:T-1}).  
\end{equation}
Since we already showed the cost is monotonically decreasing, we have 
\begin{equation}\label{eq:thm1_conv}
\frac{1}{\delta}
\sum_{k=0}^{K}
\sum_{t=0}^{T-1}
\left\|
\pi_t^{k+1}-\pi_t^{k}
\right\|_{\mathcal H}^2 \leq \tilde{J}_0(\pi^0_{0:T-1}).
\end{equation}
Since the inner summation is always nonnegative, for all $k$, and the outer summation is now monotonically non-decreasing. To ensure the double summation is bounded as \eqref{eq:thm1_conv}, from series theory, we need to have $\sum_{t=0}^{T-1} \| \pi_{t}^{k+1} - \pi_{t}^{k} \|_{\mathcal{H}}^{2} \to0$ when $K\to \infty$, and the proof is complete.
\end{proof}

\subsection{Implementation}

Since the problem lies in an infinite-dimensional functional space, we consider a finite-dimensional numerical approximation for implementation.

% In this paper, we seek approximate solutions within a linearly parameterized (possibly nonlinear) policy class.

\subsubsection{Monte Carlo Approximation}

Let $\{x_0^{(i)}\}_{i=1}^N$ be i.i.d.\ samples drawn from $p_0$.
For a given policy sequence $\pi_{0:T-1}$, forward simulation of \eqref{eq:linear_sys} yields state trajectories
$\{x_t^{(i)}\}_{t=0}^T$.

The functional $\tilde{J_t}$ is approximated by
\begin{align}\label{eq:Jt_hat}
\hat J_t(\pi_t)
=
\frac{1}{N}
\sum_{i=1}^N
&\Big[
x_t^{(i)\top} Q x_t^{(i)}
+
\pi_t(x_t^{(i)})^\top R \pi_t(x_t^{(i)})
+
\psi(x_t^{(i)})\notag\\
&
+
V_{t+1}(A x_t^{(i)} + B \pi_t(x_t^{(i)}))
\Big].
\end{align}

Let the vector of policy evaluations at the sampled states be defined as 
\begin{equation}
\boldsymbol{\pi}_t
\coloneq 
\mathrm{col}\big(
\pi_t(x_t^{(1)}),\dots,\pi_t(x_t^{(N)})
\big)
\in \mathbb{R}^{mN}.
\end{equation}

\subsubsection{RKHS-based Policy Representation}

Assume that $\pi_t$ belongs to a RKHS $\mathcal{H}$ with kernel $K(\cdot,\cdot)$.
The policy is approximated by a finite kernel expansion
\begin{equation}\label{eq:kernel_policy}
\pi_t(\cdot)
=
\sum_{j=1}^M c_{t,j} K(\cdot,\bar x_{t,j}),
\end{equation}
where $\{\bar x_{t,j}\}_{j=1}^M$ are fixed dictionary points. Let $c_t\coloneq \mathrm{col}(c_{t,1},\dots,c_{t,M})$.
Then $\boldsymbol{\pi}_t = K_{s,t} c_t$, where $(K_{s,t})_{ij} = K(x_t^{(i)},\bar x_{t,j})$.
Define the Gram matrix by
\begin{equation}
(K_t)_{ij} \coloneq  K(\bar x_{t,i},\bar x_{t,j}).
\end{equation}

The RKHS representation is used to convert the functional policy update
into a finite-dimensional coefficient update, enabling numerical solution of the
implicit learning step from sampled trajectory data.

\subsubsection{Implicit Policy Improvement}

% The policy update is performed via an implicit discretization of the Fréchet derivative
% \begin{equation}\label{eq:implicit_update}
% \boldsymbol{\pi}_t^{k+1}
% =
% \boldsymbol{\pi}_t^{k}
% -
% \delta \,
% D \hat J_t
% \big(
% \boldsymbol{\pi}_t^{k+1},
% \boldsymbol{\pi}_t^{k}
% \big),
% \end{equation}
% where $\delta>0$ is a step size.

Let
$\Delta \boldsymbol{\pi}_t \coloneq  \boldsymbol{\pi}_t^{k+1} - \boldsymbol{\pi}_t^{k}$.
Consider an approximation for the discrete Fréchet derivative given by
\begin{equation}
D \hat J_t
=
\begin{cases}
\displaystyle
\Delta \boldsymbol{\pi}_t
\frac{
\hat J_t(\boldsymbol{\pi}_t^{k+1}) - \hat J_t(\boldsymbol{\pi}_t^{k})
}{
\|\Delta \boldsymbol{\pi}_t\|_2^2
},
& \Delta \boldsymbol{\pi}_t \neq 0, \\[2ex]
0,
& \text{otherwise}.
\end{cases}
\end{equation}

Substituting $\boldsymbol{\pi}_t = K_{s,t} c_t$ yields the parameter update equation
\begin{equation}\label{eq:ck_update}
K_t c_t^{k+1}
=
K_t c_t^{k}
-
\delta K_{s,t}^\top
D \hat J_t.
\end{equation}

%{\color{blue}Computation complexity: linear QP non-linear? of \eqref{eq:ck_update}}

The update rule given by~\eqref{eq:ck_update} leads to a nonlinear optimization problem in the kernel coefficients $c_t^{k+1}$, since the empirical cost $\hat{J}_t$ depends nonlinearly on the policy output.
Under the RKHS parameterization with $m$ kernel functions and $n$ Monte Carlo samples, evaluating the policy vector and the empirical cost requires $O(mn)$ operations per time step. If a Newton-type method is employed to solve the resulting nonlinear problem, each inner iteration requires $O(m^3)$ operations due to solving an $m$-dimensional linear system.
Since the system state is formed by stacking the states of all team members, the state dimension grows linearly with $N$. Kernel evaluations, therefore, require computing distance in a space whose dimension scales with $N$, leading to a time cost $O(mnN)$. Moreover, each residual evaluation requires recomputing the forward trajectory from the current stage $t$ to the terminal horizon $T$. Consequently, the cost of evaluating the objective at each time step $t$ requires $O((T-t)nmN)$. Summing over all stages yields $O(nmNT^2)$. Therefore, the overall per-iteration complexity scales as $O\big(m^3T + nmNT^2\big)$. The computational complexity is thereby quadratic in the horizon length $T$, linear in the number of Monte Carlo samples $n$ and the number of members $N$, and cubic in the number of kernel basis functions $m$.

\subsection{Policy Iteration Algorithm}

From the preceding analysis, the resulting policy iteration algorithm can be given as follows:

\begin{algorithm}[H]
\caption{Offline learning}
\begin{algorithmic}[1]
\label{alg:offline}
\STATE Initialize $\{\pi_t^0\}_{t=0}^{T-1}$ and sample $\{x_0^{(i)}\}_{i=1}^N$
\FOR{$k=0,1,2,\dots$}
    \STATE \textbf{Forward simulation:} generate trajectories using \eqref{eq:linear_sys}
    \STATE \textbf{Backward evaluation:} compute $V_t$ for $t=T,\dots,0$
    \FOR{$t=T-1,\dots,0$}
        \STATE Solve \eqref{eq:ck_update} for $c_t^{k+1}$
    \ENDFOR
\ENDFOR
\end{algorithmic}
\end{algorithm}

Algorithm~\ref{alg:offline} implements a finite-horizon policy iteration process including policy evaluation and improvement using the Fréchet derivative. Starting from an initial sequence of policies, a forward simulation is performed using the current policy at each iteration. Subsequently, a backward policy improvement step is carried out from $t=T-1$ to $t=0$, by solving an implicit equation \eqref{eq:ck_update} derived from the Fréchet derivative. The policy returned by Algorithm~\ref{alg:offline} is an approximate local team-optimal solution of  $\pi^\star_{0:T-1}$. 
% This distinction arises due to the RKHS parameterization and with Monte Carlo approximations. Nevertheless, 
Under sufficiently rich kernel dictionaries and sampling, the learned policy can be expected to approach $\pi^\star_{0:T-1}$.

% {\color{blue}Establish the connection with team-optimal policy $\pi^\star_{0:T-1}$ I think is an approximate local team solution}

However, Algorithm~\ref{alg:offline} is implemented in an offline model-based setting and requires prior knowledge of the system. Moreover, the main computational burden of this algorithm arises from solving the non-linear implicit equation \eqref{eq:ck_update} over the entire horizon at each policy update step, as well as from the number of kernel basis functions used to construct the policy and the sampling size. 
This raises a trade-off between computational cost and approximation accuracy. In particular, the quality of the learned policy depends on both the number of Monte Carlo samples used to approximate  $\tilde{J}_t$ and the number of kernel basis functions used in the RKHS. Increasing the sample size generally reduces the variance of the empirical cost approximation, while increasing the number of kernel functions could better capture the nonlinear structure induced by the safety penalty. 
%At the same time, these gains lead to higher computational cost as shown in the complexity analysis above. 
Therefore, a sufficiently large number of sampling points and kernel basis functions can yield a policy that is closer to the team-optimal solution. However, this also makes the offline learning procedure more computationally expensive, especially for long horizons and multi-agent scenarios, as explained in the complexity analysis above. 

To address this, in what follows, we further consider an online model-free setting by reformulating the current algorithm into a receding-horizon scheme. 
% Specifically, at each time step, the policy is iteratively updated over a shorter time window. The improved policy is then implemented for the next step, after which the time window is shifted forward, and the same procedure is repeated.

\section{Model-free Online learning}

In this section, we propose the online learning algorithm under a model-free setting. We first collect online data and use recursive least squares (RLS) to estimate the system parameters $A$ and $B$, and then update the policy in a receding-horizon fashion using newly collected online data until convergence.

Consider the model-free setting in which $A$ and $B$ are unknown,
\begin{equation}\label{eq:cl_map}
x_{t+1}=A x_t + B \pi_t(x_t).
% x_{t+1}=F_{c}(x_t)\coloneq A x_t + B \pi_t(x_t),
%+ \sigma(x_t),
\end{equation}
% where $\pi^c_t:\mathbb{R}^n\to\mathbb{R}^m$ is a parameterized state-feedback policy (e.g., RKHS/kernel policy).
% %and {\color{red}$\sigma(\cdot)$} is an nonlinearity/disturbance.
% We focus on recursive stability: the closed-loop remains stable under a sequence of policy updates
% $\{c_k\}_{k\ge 0}$ produced by an online learning algorithm.

Fix a finite horizon $T$ and a receding window length $H$.
At each real time $s\in\{0,1,\dots,T-1\}$, define the planning horizon end index $\bar T_s \coloneq  \min\{T,\;s+H\}$. Given the current state $x_s$ and the model estimate $(\hat A_s,\hat B_s)$,
we consider the predicted dynamics on the remaining horizon $t=s,\dots,\bar T_s-1$:
\begin{equation}\label{eq:pred_dyn_selfcontained}
\hat x_{t+1}
=
\hat A_s \hat x_t+\hat B_s\,\pi_{s,t}(\hat x_t),
\qquad
\hat x_s=x_s.
\end{equation}
Let the remaining objective be
\begin{align}\label{eq:Jhat_selfcontained}
\tilde{J}_s(\boldsymbol \pi_s) 
& \coloneq 
\sum_{t=s}^{\bar T_s-1}
\left(
\hat x_t^\top Q \hat x_t
+
\pi_{s,t}(\hat x_t)^\top R \,\pi_{s,t}(\hat x_t)
+
\psi(\hat x_t)
\right) \notag \\
&\qquad \qquad+ 
\hat x_{\bar T_s}^\top Q_F \hat x_{\bar T_s}
+
\psi_F(\hat x_{\bar T_s}),
\end{align}
where $\boldsymbol \pi_s\coloneq \{\pi_{s,t}\}_{t=s}^{\bar T_s-1}$ and $\{\hat x_t\}$ is generated by \eqref{eq:pred_dyn_selfcontained}.
A backward implicit Fr\'echet policy improvement routine minimizes \eqref{eq:Jhat_selfcontained} by updating
$t=\bar T_s-1,\dots,s$, and returns a candidate sequence
$\tilde{\boldsymbol \pi}_s=\{\tilde \pi_{s,t}\}_{t=s}^{\bar T_s-1}$.
% Only the first policy parameter $\tilde c_{s,s}$ is applied to the real system.

The online learning scheme combines model identification and policy improvement in a receding-horizon framework. Starting from 
$s=0$, an initial excitation phase is performed to collect data and estimate the unknown system matrices using RLS. 
Once the estimation scheme terminates, 
% At each real-time step, the methodology described in section~III is applied only on the truncated horizon $[s,\bar T_s]$ using the estimated system parameters.
at each real-time step, the current system estimates are used to formulate a control problem over a short planning window. This problem is solved using the methodology described in Section III.
Only the first policy in the resulting sequence is implemented on the real system, after which the planning window is shifted forward, and the procedure is repeated. In this way, the algorithm yields a receding-horizon feedback policy based on the identified model, thereby avoiding the computational burden of solving the full-horizon problem offline.

% We first use system ID to estimate  in a time period and then solving the sequence receding-horizon problems ........

\noindent\paragraph{Linear regression formulation}
We first describe the model identification process \cite{ioannou1996robust}. 
Define the parameter matrix
$\Theta \coloneq  [A \;\; B] \in \mathbb{R}^{n \times (n+m)}$,
and the regressor vector
$\phi_s \coloneq 
\begin{bmatrix}
x_s \\
u_s
\end{bmatrix}
\in \mathbb{R}^{n+m}$.
Then the system dynamics can be written as a linear regression model
$x_{s+1} = \Theta \phi_s$.

Let $\hat\Theta_0$ and $M_0 \succ 0$ be given.
At each time step $s \ge 0$, the parameter estimate is updated recursively using the newly collected data triple $(x_s,u_s,x_{s+1})$.
\begin{align}\label{eqRLS}
\phi_s &= [x_s;u_s], \\
L_s &= \frac{M_s\phi_s}{1+\phi_s^\top M_s\phi_s}, \notag\\
\varepsilon_s &= x_{s+1}-\hat\Theta_s\phi_s, \notag\\
\hat\Theta_{s+1} &= \hat\Theta_s+\varepsilon_s L_s^\top, \notag\\
M_{s+1} &= \frac{1}{\lambda}\big(M_s - L_s\phi_s^\top M_s \big), \notag
\end{align}
% \begin{align}\label{RLS}
% \phi_k &=
% [x_k ; u_k]
%  \notag\\
% L_k &= \frac{M_k \phi_k}{\lambda + \phi_k^\top M_k \phi_k}, \notag\\
% \varepsilon_k &= x_{k+1} - \hat\Theta_k \phi_k, \\
% \hat\Theta_{k+1} &= \hat\Theta_k + \varepsilon_k L_k^\top, \notag\\
% M_{k+1} &= \frac{1}{\lambda}
% \left(M_k - L_k \phi_k^\top M_k\right). \notag
% \end{align}
where $\lambda \in (0,1]$ is the forgetting factor.
The matrices $\hat A_s$ and $\hat B_s$ are obtained from $\hat\Theta_s = [\hat A_s \;\; \hat B_s]$.

% \begin{assumption}\label{est_error}
% At each policy iteration $k$, an online identification routine returns $(\hat A_k,\hat B_k)$ such that
% \begin{equation}\label{eq:est_err_bounds}
% \|A-\hat A_k\|\le \varepsilon_{A,k},\qquad \|B-\hat B_k\|\le \varepsilon_{B,k},
% \end{equation}
% where $\varepsilon_{A,k},\varepsilon_{B,k}\ge 0$ are known bounds.
% \end{assumption}

\begin{assumption}\label{pe}
The regressor sequence $\{\phi_k\}$ is persistently exciting of order $n+m$, i.e., 
% there exist $N_{\mathrm{PE}}\in\mathbb{N}$ and $\alpha>0$ such that for all
% $k\ge N_{\mathrm{PE}}$,
% \[
% \sum_{i=k-N_{\mathrm{PE}}}^{k-1}\phi_i\phi_i^\top \succeq \alpha I_{n+m}.
% \]
there exist constants
$N_{\mathrm{PE}} > 0$ and $\alpha > 0$ such that for all $s \ge N_{\mathrm{PE}}$,
\begin{equation}\label{eq:PE}
\sum_{i=s-N_{\mathrm{PE}}}^{s-1}
\phi_i \phi_i^\top
\succeq
\alpha I_{n+m}.
 \end{equation}
\end{assumption}

\begin{algorithm}[t]
\caption{Online learning}
\label{alg:online_backward}
\begin{algorithmic}[1]
\STATE{Horizon $T$, window length $H$ (set $H=T$ for full remaining horizon), identification length $N_{\mathrm{id}}$, stability margin $\gamma\in(0,1)$, initial $c_{0,0}\in\mathcal S_0$.}
{Collect an initial dataset to start RLS; set $s=0$.}

\hspace{-0.4cm}\textbf{Model identification:}

\STATE Apply an excitation policy for $s=0,\dots,N_{\mathrm{id}}-1$.
\FOR{$s=0,1,\dots,N_{\mathrm{id}}-1$}
\STATE Observe $x_{s+1}$ from the system.
\STATE Update $(\hat A_s,\hat B_s)$ using the RLS update \eqref{eqRLS} with the newly observed data $(x_s,u_s,x_{s+1})$.
\ENDFOR

\STATE Use the estimated $(\hat{A}, \hat{B})$ for policy updating.

\hspace{-0.4cm}\textbf{Policy Update:}
\FOR {$s=N_{\mathrm{id}},N_{\mathrm{id}}+1,\dots,T-1$}{
% \STATE \textbf{(RLS model update)}
%\STATE \textbf{(RLS model update on real data)} 
%Update $(\hat A_k,\hat B_k)$ using all data $\{(x_\tau,u_\tau,x_{\tau+1})\}_{\tau=0}^{k-1}$,
%and obtain bounds $(\varepsilon_{A,k},\varepsilon_{B,k})$ as in \eqref{eq:est_err_bounds}.\\

\STATE %\textbf{(Define remaining horizon)} 
Set $\bar T_s\coloneq \min\{T,s+H\}$. 

\STATE Solve %(approximately) 
the estimated remaining-horizon problem
\[
\min_{\mathbf c_s=\{c_{s,t}\}_{t=s}^{\bar T_s-1}} \ \hat J_s(\mathbf c_s)
\]
% where $\hat J_k$ is defined by \eqref{eq:Jhat_def} and 
% \[
% \hat x_{t+1}=\hat A_k \hat x_t+\hat B_k\,\pi_{c_{k,t}}(\hat x_t),
% \hat x_k=x_k,
% t=k,\dots,\bar T_k-1.
% \]
The backward implicit Fr\'echet routine updates $t=\bar T_s-1,\dots,s$ and returns a candidate sequence
$\tilde{\mathbf c}_s=\{\tilde c_{s,t}\}_{t=s}^{\bar T_s-1}$.\\
% \STATE
% \textbf{(Safe execution at the current time)} Let $\tilde c_{k}\coloneq \tilde c_{k,k}$ be the candidate parameter to be executed now.
% Compute the maximal safe step size
% \[
% \alpha_k\coloneq \max\Big\{\alpha\in[0,1]:\ c_{k}^{\mathrm{exe}}+\alpha(\tilde c_k-c_{k}^{\mathrm{exe}})\in\mathcal S_k\Big\},
% \]
% and set the executed parameter
% \[
% c_{k}^{\mathrm{exe},+}\coloneq c_{k}^{\mathrm{exe}}+\alpha_k(\tilde c_k-c_{k}^{\mathrm{exe}}).
% \]
\STATE Apply
$
u_s=\pi^{\tilde{c}}_{{s},s}(x_s)$,
and observe  $x_{s+1}.$

\STATE 
%\textbf{(Warm-start for the next planning step)} 
%Optionally 
Set the next planning initializer by shifting:
$c_{s+1,t}^{(0)} \leftarrow \tilde c_{s,t}$ for $t=s+1,\dots,\bar T_s-1$.
}
\ENDFOR
\end{algorithmic}
\end{algorithm}

\begin{lemma}
Under Assumption~\ref{pe}, $(\hat{A}_s, \hat{B}_s)$ obtained by the RLS estimation converges to the true  $(A,B)$.  
\end{lemma}

\begin{proof}
Since the system satisfies $x_{s+1}=\Theta\phi_s$, and using Assumption~\ref{pe}, we define the parameter estimation error as $\tilde{\Theta}_s\coloneq \Theta-\hat{\Theta}_s$. The a priori prediction error becomes $\epsilon_s = x_{s+1}-\hat{\Theta}_s\phi_s=\tilde{\Theta}_s\phi_s$. Substituting this into the RLS update in \eqref{eqRLS}, the error recursion $\tilde{\Theta}_{s+1}=\tilde{\Theta}_s-\tilde{\Theta}_s\phi_s L_s^\top=\tilde{\Theta}_s(I-\phi_s L_s^\top)$. 
Assume $\lambda=1$. From the matrix inversion lemma, the covariance update in \eqref{eqRLS} is equivalent to $M_{s+1}^{-1}=M_s^{-1}+\phi_s\phi_s^\top$, and thus $M_s^{-1}=M_0^{-1}+\sum_{i=0}^{s-1}\phi_i\phi_i^\top$. 
Under the PE condition, there exists constants $N_{PE}>0$ and $\alpha>0$ such that $\sum_{i=s-N_{PE}}^{s-1}\phi_i\phi_i^\top \succeq \alpha I_{n+m}$ for all sufficiently large $s$. Accumulating these uniformly positive definite terms over time implies that $\lambda_{\min}(M_s^{-1}) \to \infty$ as $s \to \infty$, and hence $\|M_s\| \to 0$. As established in the standard convergence properties of the RLS algorithm, see \cite{ioannou1996robust}, when the regressor sequence satisfies PE, the covariance matrix $M_s$ converges to $0$, which guarantees that the parameter estimation error also converges to $0$. Therefore, $\hat{\Theta}_s\to\Theta$, implying $(\hat{A}_s, \hat{B}_s) \to (A,B)$.
% We then consider the Lyapunov function $V_s=\operatorname{tr}(\tilde{\Theta}_s M_s^{-1}\tilde{\Theta}_s^\top)$ which is nonnegative and non-increasing under the RLS recursions. Under the PE condition, $V_s \to 0$, and hence $\tilde{\Theta}_s \to 0$ as $s \to \infty$. Consequently, $
\end{proof}

Based on the identified $\hat{A}$, $\hat{B}$, we then seek the policy sequence over each rolling window based on the discrete Fréchet derivative. 
%T
%he following results shows the convergences of Alg

%{\color{blue}  and in Theorem 2, first say something  that  the policy update converges in each windows, and then say the algorithm~2 under the whole horizon converges and the state is bounded  }

% At time $k$, the planning step produces a time-varying parameter sequence
% $\{c_{k,t}\}_{t=k}^{\bar T_k-1}$.
% Following the receding-horizon principle, only the first parameter
% $c_{k,k}$ is applied to the real system.
% The optimization is repeated at the next time step.

\begin{mythm}\label{online_bounded}
%Consider the system in \eqref{eq:linear_sys_2} for $t=0,1,\dots,T-1$ over a finite horizon $T<\infty$, and policy $\pi_c(x_t)$ is associated to the RKHS expansion. Assume the system is operating under $\pi_t(x_t)$. 
For $t=0,1,\dots,T-1$, the Fréchet policy update \eqref{eq:implicit_update} converges in each time window $[t,t+H]$ and  $\|x_t\|$ is bounded. 
% the receding-horizon problem admits a local minimum. Moreover, $\|x_t\|$ is bounded and Algorithm~2 converges on the entire horizon.
\end{mythm}

\begin{proof}
At each time $s$ in the window $H$, Algorithm~\ref{alg:online_backward} solves the estimated finite-horizon problem to minimize \eqref{eq:Jhat_selfcontained} within $\bar{T}_s=\min\{T,s+H\}$. We denote $\boldsymbol \pi^0_s=\{\pi^0_{s,t}\}^{\bar T_s-1}_{t=s}$ be the initial policy sequence for the window at time $s$, and $\tilde{\boldsymbol \pi}_s=\{\tilde{\pi}_{s,t}\}^{\bar T_s-1}_{t=s}$ be the sequence returned by the backward implicit Fréchet policy improvement over $t=\bar T_s, \dots, s$. Using the descent argument in the proof of Theorem 1, we can write $\hat J_s(\tilde{\boldsymbol{\pi}}_s)\leq \hat J_s(\boldsymbol{\pi^0_s})$. 
Similar to Theorem 1, the backward update produces a decreasing window cost relative to the initial policy sequence.
Since the rolling time index $s$ takes only finite values in $\{0,1,\dots,T-1\}$, the argument above applies to every time window encountered by Algorithm~\ref{alg:online_backward}. 

Under Lemma 1, For each time step $s$, using the system definition in \eqref{eq:linear_sys_2}, we can write $x_s=A^sx_0 + \sum^{s-1}_{j=0}A^{s-1-j}B\pi_j(x_j)$. Then, by the triangle inequality, we have $\|x_s\|\leq \|A^s\|\|x_0\| + \sum^{s-1}_{j=0}\|A^{s-1-j}B\|\|\pi_j(x_j)\|$. Then we need to show that $\|\pi_j(x_j)\|$ is bounded for each $j$ in the summation. From each stage cost, we have $C_J\|\pi_s(x_s)\|^2 \leq \pi_s(x_s)^\top R\pi_s(x_s)$, where $C_J\coloneq \lambda_{\min}(R)$ is a constant. 
Moreover, by the monotonically decreasing window cost shown previously, since the horizon is finite,  we can write $C_J\sum^{T-1}_{s=0}\|\pi_s(x_s)\|^2\leq \tilde{J}_s(\tilde{\boldsymbol{\pi}}_s)\leq \tilde J_s(\boldsymbol{\pi^0_s})$. Hence, $\sum^{T-1}_{s=0}\|\pi_s(x_s)\|^2 \leq \frac{J_0}{C_J}$. Therefore, for each time $s$, we have $\|\pi_s(x_s)\|\leq \sqrt{\frac{J_s(\boldsymbol{\pi^0_s})}{C_J}}$, indicating that the $\|\pi_j(x_j)\|$ is bounded in previous inequality, resulting in a bounded $\|x_s\|$ for each time $s$, which concludes the proof.
% backup: 
% For each fixed window, the estimated objective is non-increasing under the
% backward implicit update, so the returned control sequence has cost no larger
% than that of the initializer. Hence the control inputs generated within the
% window are bounded by the corresponding finite window cost. Since only the
% first input is implemented and the overall horizon is finite, the applied input
% sequence $\{u_s\}_{s=0}^{T-1}$ is bounded. It follows from the linear dynamics
% $x_{s+1}=Ax_s+Bu_s$ that $\{x_s\}_{s=0}^{T}$ is bounded.
\end{proof}

The online learning scheme is summarized in Algorithm~\ref{alg:online_backward}. 
%It combines model identification and policy improvement in a receding-horizon framework.
An initial excitation phase is first used to collect data and estimate $(A,B)$ through RLS. Once the estimation scheme converges, then, at each real-time step, a control problem over a short time window is formulated based on the current system estimates. Similar to Algorithm~\ref{alg:offline}, Algorithm~\ref{alg:online_backward} does not compute the policy in the infinite-dimensional space. Instead, it solves the approximate remaining-horizon problem over $[s,\bar T_s]$ using the estimated model and an RKHS parameterization to obtain $\pi^{\tilde{c}}_t(x_t)$ for $t=s,\dots,\bar{T}_s-1$, where the backward implicit Fréchet update yields $\hat{J}$ as an approximation of $\tilde{J}$.
%In particular, the policy at each stage is represented by $\pi^c_s(x_s)$ using the kernel expansion with the corresponding kernel coefficient sequence $c$ over the current window. 
The first policy in this sequence is then implemented on the system, after which the planning window is shifted forward, and the procedure is repeated.

% Similar to Algorithm~\ref{alg:offline}, Algorithm~\ref{alg:online_backward} does not compute the policy in the infinite dimensional space....
% %the exact team optimal policy $\pi^\star_{0:T-1}$ directly.
%  Instead, it solves the approximate remaining horizon problem over the interval from $s$ to $\bar T_s$ using the current model estimate and the RKHS parameterization. In particular, the policy at each stage is represented by $\pi^c_s(x_s)$ using the kernel expansion with the corresponding kernel coefficient sequence $c$ over the current window.  Thus, 

% In Algorithm~\ref{alg:online_backward}, the receding-horizon update reduces the computation at each time step by solving the policy improvement problem only over the shorter time window~$H$, instead of over the full remaining horizon.
%while the system identification stage removes the requirement for prior knowledge of the system dynamics.
%Under the same sample size and kernel basis functions, using a smaller window length $H$ generally lowers the computational complexity. 
However, in Algorithm~\ref{alg:online_backward}, since the computational cost is saved by optimizing over a truncated horizon, the resulting receding-horizon policy is a suboptimal solution relative to the full horizon problem. Clearly, different choices of the horizon length $H$ lead to a trade-off between solution accuracy and computational cost.

% {\color{blue}establish the connection with $\pi^\star_{0:T-1}$, like a sub-optimal team solution....
% Similarly to the implementation in Alg~1, we use approximation here....... summarized in Alg~2
% Consider using $\pi^c_t$, $\pi_{0:T-1}^c$
% }

\section{Simulation}

In this section, we evaluate the proposed method on a signal-free intersection coordination problem under both offline and online settings.

In the offline setting, we consider a signal-free intersection with two CAVs approaching from orthogonal directions. Each vehicle is assigned a fixed path through the intersection, without turning or changing lanes. Each vehicle is modeled by the double integrator in \eqref{eq:vehicle_dynamic}, where the state of vehicle $i$ consists of its position $p_i$ and its speed $v_i$, and the control input is its acceleration. The control objective \eqref{eq:objective_intersection} is to coordinate their longitudinal motions so both vehicles cross the conflict area safely and efficiently. In the online setting, we consider a mixed traffic extension scenario with two CAVs and one  HDV. CAV-1 travels from west to east, HDV travels from east to west, and CAV-2 travels from south to north. The two CAVs remain controlled, whereas the HDV is human-driven with behavior unknown to the CAVs, resulting in unknown system dynamics.

% We first illustrate how the offline algorithm can learn an approximate team-optimal feedback policy when the system dynamics are known. Then, we show how the online receding-horizon algorithm can obtain the policy when the system is unknown.  
In the numerical implementation, the expectation in the policy improvement functional is approximated by Monte Carlo sampling. The vehicles are initialized upstream of the intersection region, with initial positions drawn so that both vehicles must start to adjust their speed before entering the conflict area. In both simulations, we consider the intersection length to be 10 meters. The dynamics are discretized by $\Delta t=0.1s$ in both settings.

\begin{figure}[h]
    \centering
    \includegraphics[width=\linewidth]{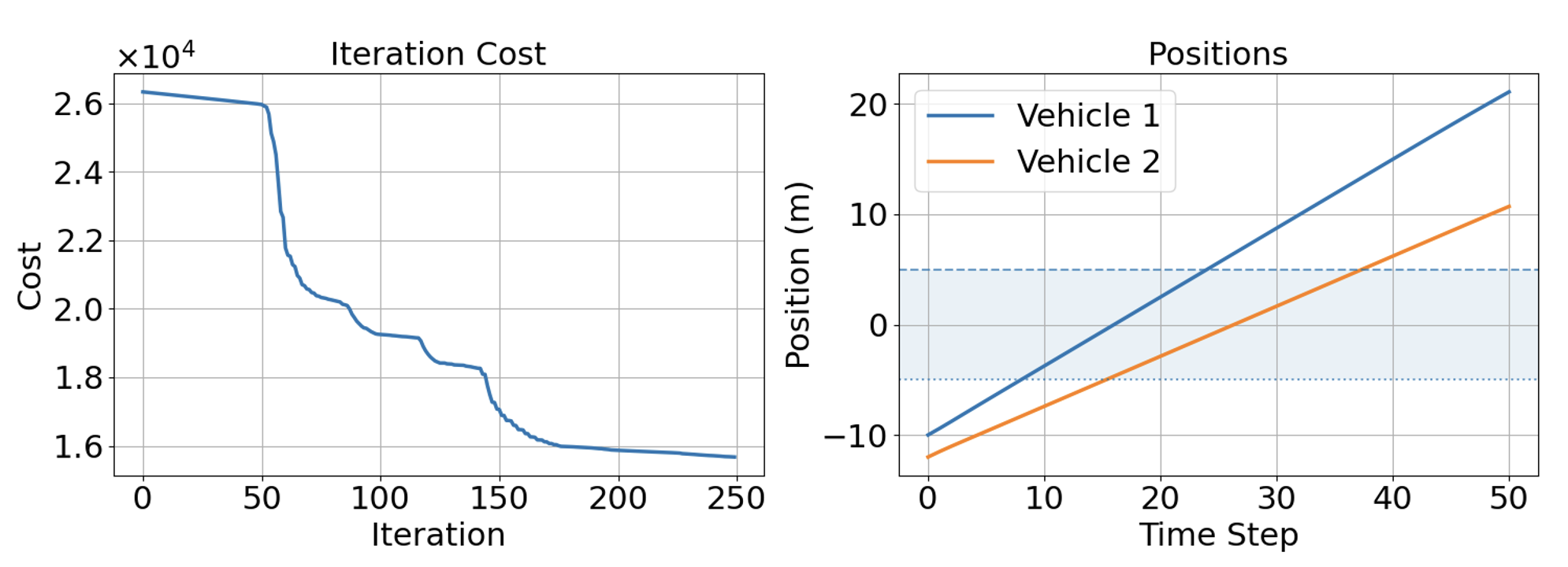}
    \caption{Offline cost (left) and vehicle positions (right).}
    \label{fig:offline_1}
\end{figure}

\begin{figure}[h]
    \centering
    \includegraphics[width=\linewidth]{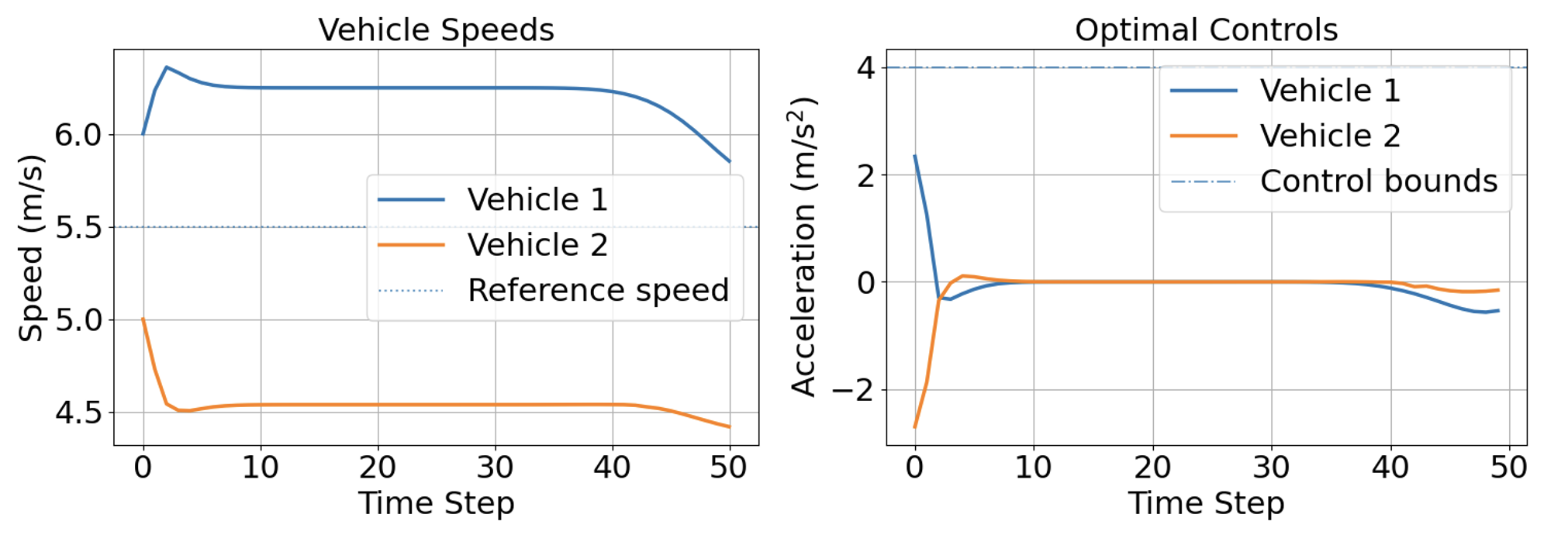}
    \caption{Vehicle speeds (left) and controlled acceleration (right).}
    \label{fig:offline_2}
\end{figure}

\begin{figure}[h]
    \centering
    \includegraphics[width=\linewidth]{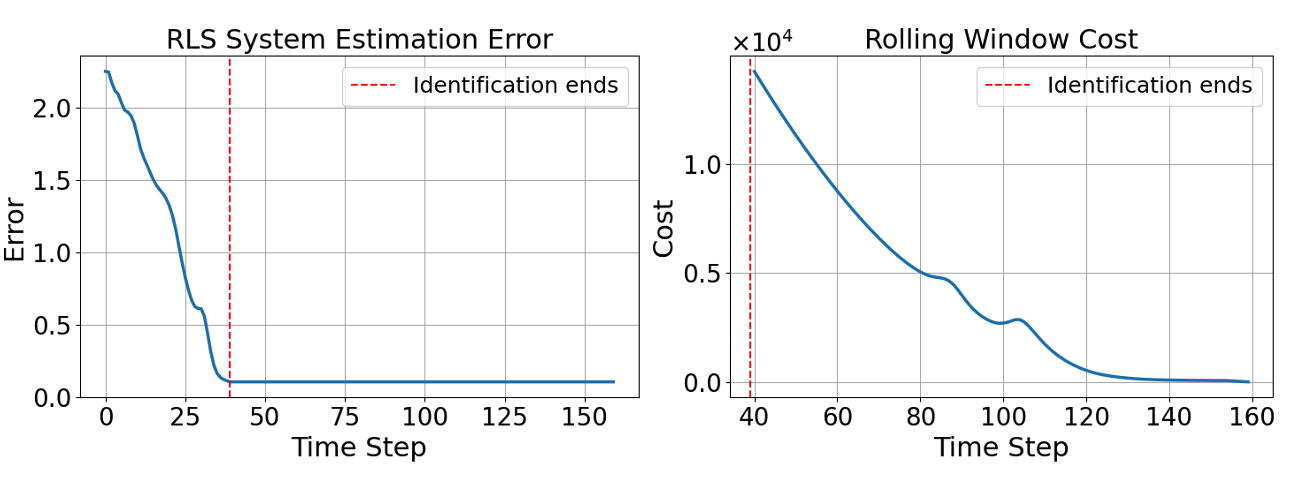}
    \caption{RLS identification error (left) rolling window cost (right).}
    \label{fig:online_1}
\end{figure}

\begin{figure}[h]
    \centering
    \includegraphics[width=\linewidth]{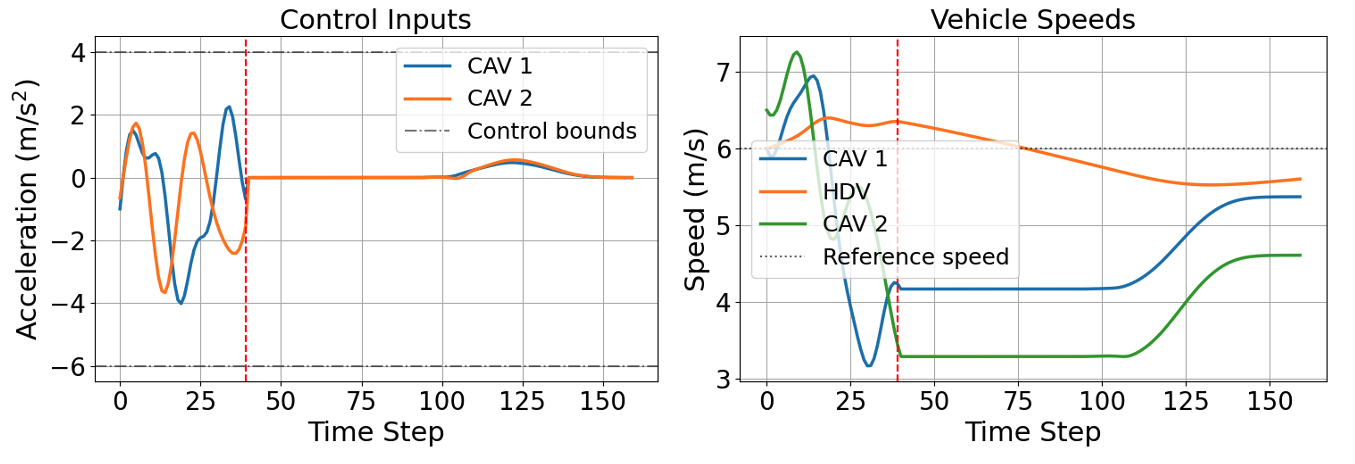}
    \caption{CAV control inputs (left) and vehicle speeds (right).}
    \label{fig:online_2}
    \vspace{-0.5cm}
\end{figure}

In the offline algorithm simulation, we set the full horizon to be $T=50$. The system model is assumed available, and the policy is computed over the entire finite horizon around $250$ iterations. We then test the trained policy on a scenario where two vehicles start at a similar distance to the conflict area at close speeds. The left figure in Fig.~\ref{fig:offline_1} shows that the cost is overall decreasing across iterations and converges at around $200$ iteration. The right figure shows that two vehicles passing through the intersection with different speeds without conflicts, with Vehicle~1 crossing earlier and Vehicle~2 following after sufficient separation. In Fig.~\ref{fig:offline_2}, it shows that both vehicles begin adjusting their speeds before entering the conflict region and cruising through the intersection without aggressive behavior.

In the online simulation, 
%we consider a mixed traffic intersection scenario by adding one additional HDV. 
the objective remains the same to control the two CAVs, but with the presence of the HDV. The HDV does not follow a prescribed model available to the controller due to unknown human driver behaviors. Instead, it reacts to the surrounding vehicles through an unknown human driving response. Consequently, the system parameters $A$ and $B$ are no longer available and must be identified from data.
% In the online simulation, the same intersection scenario and cost structure are retained. However, the system model is not assumed to be known a priori and is estimated from the data during the identification phase, followed by policy updates in a receding-horizon manner. 
During the RLS phase, the control inputs are persistently exciting with a Gaussian noise with $\mu=0$ and $\sigma=1.5$ applied.
%, which produces noticeable fluctuations in both speed and acceleration.
These random oscillations are reflected in the vehicle speed and controls in Fig.~\ref{fig:online_2} before the identification stops.
% \tp{Can we add diagrams of the learned policies as functions of the state?}
During this phase, the initial policy is taken as $\pi_0(x_0)=0$ plus the noise. In this scenario, the identification phase ends at $t=40$, as shown by the left figure in Fig.~\ref{fig:online_1}, indicating the rapid decay of the RLS estimation error decreases close to $0$. Once the identification is completed, the noise injection into the acceleration is stopped, and the CAVs accelerate in a coordinated way in the presence of the HDV. The rolling window's length is set to $\bar{T}_s=4$. The learned policies are further plotted as functions of the states in Fig.~\ref{fig:online_4}.
%The rolling window objective $\hat{J}_s$ decreases sharply after the identification phase. 
The position trajectories in Fig.~\ref{fig:online_3} show that all vehicles start well upstream of the intersection, enter the conflict region after the identification stage, and cross it smoothly with appropriate spacing. Furthermore, the pairwise distance in Fig~\ref{fig:online_3} shows the separations between conflict vehicle pairs. The distances remain above the prescribed safety threshold.
% In addition, the fixed window inner cost shows a decreasing pattern, which is consistent with local improvement within each online update, as illustrated in Fig.~\ref{fig:online_3}.

\begin{figure}[h]
    \centering
    \includegraphics[width=0.95\linewidth]{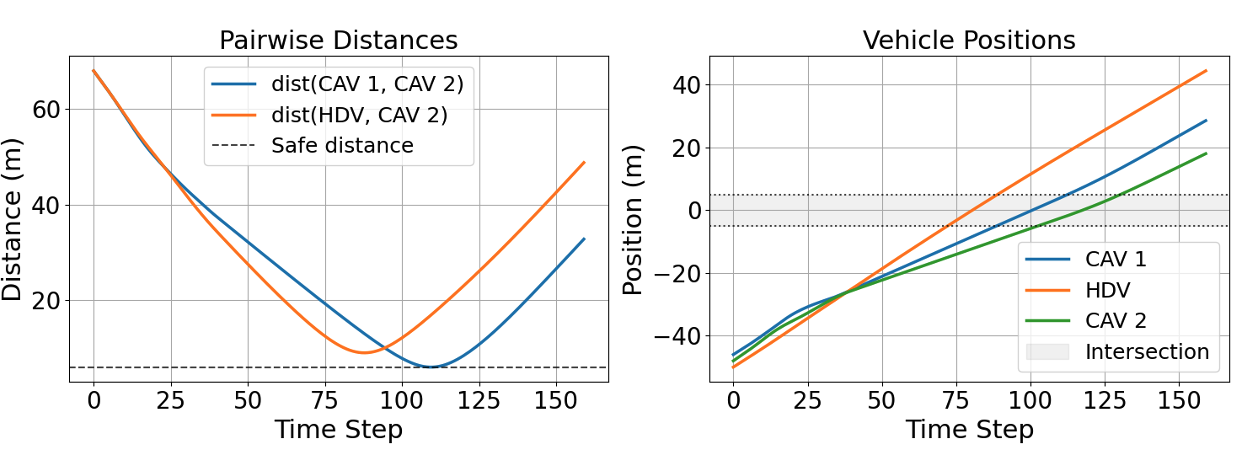}
    \caption{Pairwise vehicle distances (left) and trajectories (right).}
    \label{fig:online_3}
\end{figure}

To validate our results, we run a SUMO simulation based on the optimal controllers from the online algorithms with the HDV present. As shown in Fig.~\ref{fig:sumo}.(a), both CAVs enter the intersection without any learned controller applied, which leads to a collision at the center of the intersection, while Fig.~\ref{fig:sumo}.(b) shows that by applying the proposed online algorithm, CAVs could pass through coordinately without sacrificing efficiency on waiting. Although this comparison is intended only as a simple qualitative demonstration, it provides additional evidence that the proposed framework can be transferred to a larger-scale vehicle coordination setting. 
\begin{figure}[t]
    \centering
    \includegraphics[width=0.95\linewidth]{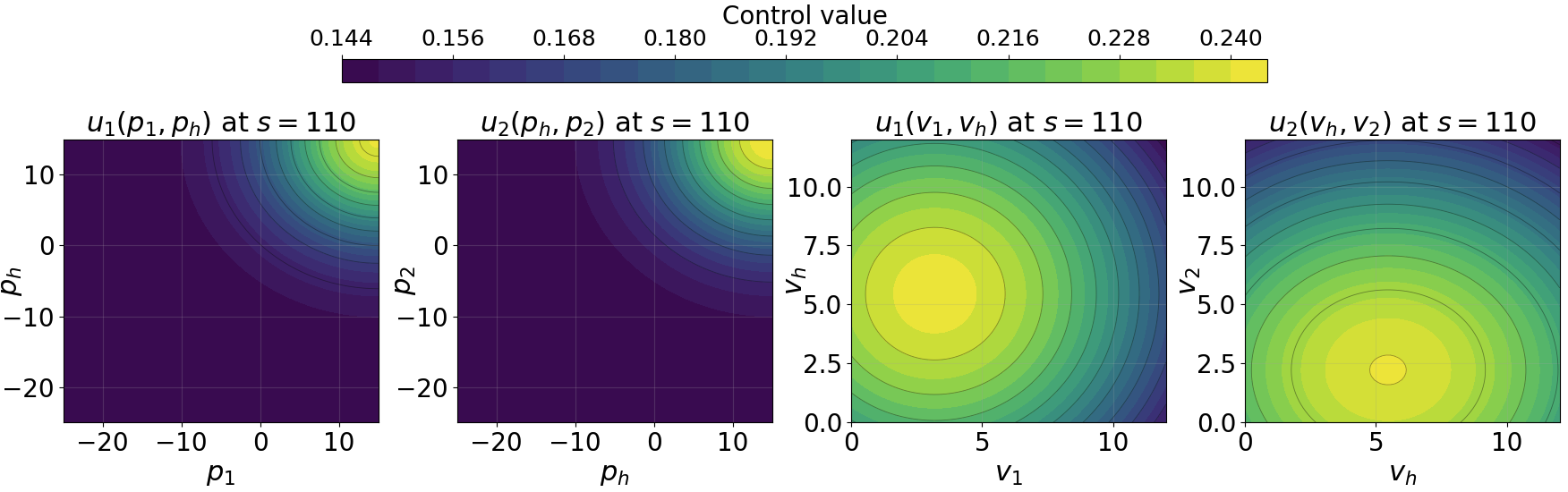}
    \caption{Learned policy as functions of the states.}
    \label{fig:online_4}
\end{figure}

\begin{figure}[]
    \centering
    \includegraphics[width=0.95\linewidth]{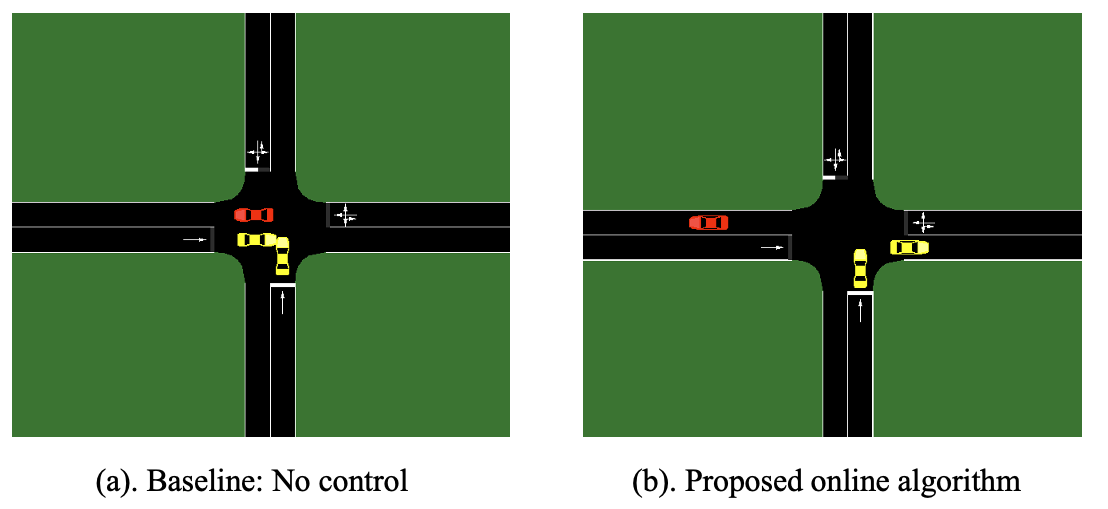}
    \caption{Sumo scenarios with CAVs (yellow) and HDV (red).}
    \label{fig:sumo}
    \vspace{-0.3cm}
\end{figure}

\section{Concluding Remarks}
In this paper, we developed a functional learning approach for finite-horizon team-optimal coordination in signal-free traffic scenarios under both offline and online settings. The policy was updated via an implicit discrete Fréchet derivative in an RKHS. In the offline setting, updates were performed over the full horizon using sampled trajectories under known dynamics, whereas in the online setting, unknown dynamics were recursively estimated via RLS, followed by receding-horizon policy updates to reduce computational burden. Simulation results demonstrated the effectiveness of the proposed method.

Future work should focus on further reducing computational burden, extending the framework to constrained and safety-critical coordination problems, and validating the proposed algorithms on large-scale traffic scenarios.

\section{Acknowledgment}

The authors wish to thank Leon Khalyavin for providing the generic code of the functional learning tool-set that has been customized to the application at hand in this paper.

\bibliographystyle{IEEEtran}
\bibliography{reference,ids}

\end{document}